\documentclass[11pt]{article}
\usepackage{amssymb,amsmath}
\usepackage{amsthm}
\usepackage{graphicx}
\usepackage[small]{caption}
\usepackage{subcaption}
\usepackage{epsfig}
\usepackage{amsfonts}

\usepackage[utf8]{inputenc}
\usepackage{fullpage}

\usepackage{enumerate}
\usepackage{url}
\usepackage{hyperref}
\usepackage{complexity}

\newcommand{\old}[1]{{}}
\newcommand{\later}[1]{{}}

\newtheorem{theorem}{Theorem}
\newtheorem{lemma}{Lemma}

\newtheorem{corollary}{Corollary}

\newtheorem{observation}{Observation}

\def\etal{{et~al.}}
\def\eg{{e.g.}}
\def\ie{{i.e.}}

\newcommand{\eps}{\varepsilon}

\newcommand{\RR}{\mathbb{R}}

\def\F{\mathcal F}

\def\V{\mathcal V}

\newcommand{\alg}{\textsf{ALG}}
\newcommand{\opt}{\textsf{OPT}}
\newcommand{\orp}{\textsf{ORP}}
\newcommand{\ewrp}{\textsf{EWRP}}
\newcommand{\tspn}{\textsf{TSPN}}
\newcommand{\etsp}{\textsf{ETSP}}
\newcommand{\wrp}{\textsf{WRP}}
\newcommand{\rtsp}{\textsf{RTSP}}
\newcommand{\tsp}{\textsf{TSP}}
\newcommand{\setcov}{\textsf{Set Cover}}

\newcommand{\len}{{\rm len}}
\newcommand{\per}{{\rm per}}
\newcommand{\diam}{{\rm diam}}
\newcommand{\conv}{{\rm conv}}

\begin{document}
\title{Observation Routes and External Watchman Routes}
\author{Adrian Dumitrescu\thanks{Algoresearch L.L.C., Milwaukee, WI, USA.
Email: \texttt{ad.dumitrescu@algoresearch.org}}
\and
Csaba D. T\'oth\thanks{Department of Mathematics, California State University Northridge, Los Angeles, CA; and Department of Computer Science, Tufts University, Medford, MA, USA. Email: \texttt{csaba.toth@csun.edu}. Research on this paper was supported, in part, by the NSF awards DMS-1800734 and DMS-2154347.}
}
\date{}
\maketitle              % typeset the header of the contribution

\begin{abstract}
  We introduce the Observation Route Problem (\orp) defined as follows: Given a set of $n$ pairwise disjoint
  compact regions in the plane, find a shortest tour (route) such that an observer walking along this tour can see
  (observe) some point in each region from some point of the tour. The observer does \emph{not} need to see
  the entire boundary of an object. The tour is \emph{not} allowed to intersect the interior of any region
  (\ie, the regions are obstacles and therefore out of bounds).
  The problem exhibits similarity to both the Traveling Salesman Problem with Neighborhoods (\tspn)
  and the External Watchman Route Problem (\ewrp).
  We distinguish two variants: the range of visibility is either limited to a bounding rectangle, or unlimited.
We obtain the following results:

\smallskip
(I) Given a family of $n$ disjoint convex bodies in the plane, computing a shortest observation route does not admit a $(c\log n)$-approximation unless $\P = \NP$ for an absolute constant $c>0$. (This holds for both limited and unlimited vision.)

\smallskip
(II) Given a family of disjoint convex bodies in the plane, computing a shortest external watchman route is $\NP$-hard. (This holds for both limited and unlimited vision; and even for families of axis-aligned squares.)

\smallskip
(III) Given a family of $n$ disjoint fat convex polygons, an observation tour
whose length is at most $O(\log{n})$ times the optimal can be computed in polynomial time. (This holds for limited vision.)

\smallskip
(IV) For every $n \geq 5$, there exists a convex polygon with $n$ sides and all angles obtuse such that its perimeter is \emph{not} a shortest external watchman route.  This refutes a conjecture by Absar and Whitesides (2006).
\end{abstract}

\section{Introduction}  \label{sec:intro}

Path planning and visibility are two central areas in computational geometry and robotics.
In path planning, a short collision-free path between two specified points is desired,
and the robot has to see or detect obstacles in order to avoid them in its path. Hence there is a close relation
between short paths and visibility. Moreover, visibility of an object (say, an obstacle) can be accomplished at
various degrees; for instance, sometimes it may suffice to simply detect the presence of an obstacle, and other times the robot may need to map or recognize (\eg, see) the entire boundary of an obstacle in order to select a meaningful action.

In the Traveling Salesman with Neighborhoods problem ($\tspn$), given a set of regions (neighborhoods)
in the plane, one is to compute a shortest closed route (tour)
that visits each neighborhood; whereas in the External Watchman Route Problem ($\ewrp$), given a set of disjoint regions in the plane, one is to compute a shortest closed route (tour) in the exterior of a region (\ie, in the \emph{free space}) so that every point on the boundary of every region is visible from some point of the tour. These problems were posed about three decades ago by Arkin and Hassin~\cite{AH94}
  and by Ntafos and Gewali~\cite{NG94}, respectively.
A small example that illustrates $\orp$ and $\ewrp$ appears in Fig.~\ref{fig:star}.
Here we introduce the following related problem we call the Observation Route Problem (\orp):

\begin{quote}
\orp: Given a set of $n$ pairwise disjoint compact regions in the plane, find a shortest route (tour) such that an observer going along this tour can see (observe) each of the regions from at least one point of the tour. The tour cannot enter the interior of any region.
\end{quote}

\begin{figure}[ht]
\centering
\includegraphics[scale=0.7]{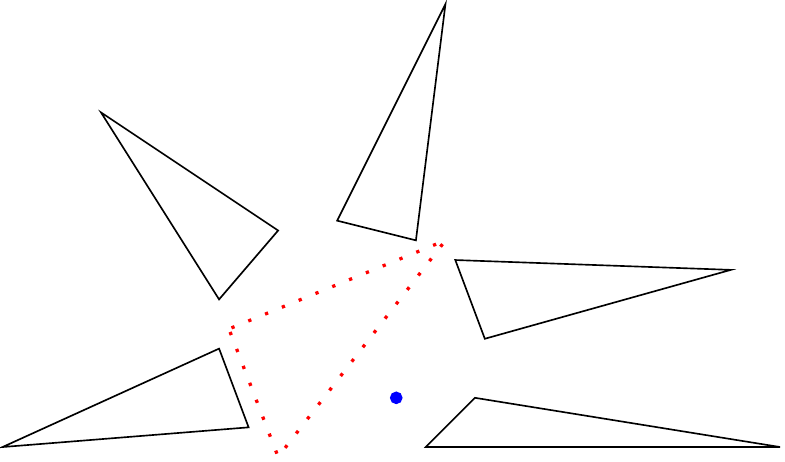}
\caption{An observation route (the blue point) and an external watchman route (dotted, in red)
  for a set of five triangles.}
\label{fig:star}
\end{figure}

\paragraph{Related work.}
In the Euclidean Traveling Salesman problem (\etsp), given a set
of points in the plane (or in the Euclidean space $\RR^d$, $d \geq 3$),
one seeks a shortest tour (closed curve) that visits each point.
In the \emph{TSP with neighborhoods} (\tspn), each point is replaced by a
(possibly disconnected) region~\cite{AH94}. The tour must visit
at least one point in each of the given regions (\ie, it must intersect each region).
Since $\etsp$ is $\NP$-hard in $\RR^d$ for every $d \geq 2$ \cite{GGJ76,GJ79,P77}, $\tspn$ is also $\NP$-hard for every $d \geq 2$.

At about the same time, Arora~\cite{Ar98} and Mitchell~\cite{Mi99}
independently showed that $\etsp$ in $\RR^d$, for constant $d$,
admits a polynomial-time approximation scheme (PTAS).
In contrast, $\tspn$ is harder to approximate and certain instances are known to be $\APX$-hard. However, better approximations can be obtained for neighborhoods with ``nice'' geometric properties: connected, pairwise disjoint, or fat, or of comparable sizes, etc. %Some results are briefly reviewed below.
Arkin and Hassin~\cite{AH94} gave constant-factor approximations for translates of a connected region; Dumitrescu and Mitchell~\cite{DM03} extended the above result to connected neighborhoods of comparable diameters.

For $n$ connected (possibly overlapping) neighborhoods in the plane,
$\tspn$ can be approximated with ratio $O(\log{n})$ by the algorithms of
(i)~Mata and Mitchell~\cite{MM95},
(ii)~Gudmundsson and Levcopoulos~\cite{GL99}, and
(iii)~Elbassioni, Fishkin, and Sitters~\cite{EFS09}.
The $O(\log{n})$-approximation stems from the following early result by Levcopoulos and Lingas~\cite{LL84}:
Every (simple) rectilinear polygon $P$ with $n$ vertices, $r$ of which are reflex, can be partitioned
in $O(n \log{n})$ time into rectangles whose total perimeter is $\log{r}$ times the perimeter of $P$.
We will use any of the three algorithms mentioned above as a subroutine in our approximation algorithm
for $\orp$ in Section~\ref{sec:approx}.

In the Watchman Route Problem ($\wrp$), given a polygonal domain $P$,
the goal is to find a shortest closed curve within $P$ such that every point of $P$
is seen from some point along the curve.
Thus $\wrp$ is dual to $\ewrp$ in the sense that the former deals with the interior of
a polygonal domain whereas the latter deals with the exterior of one or more polygons.
The watchman route problem in a simple polygon $P$
first considered by Chin and Ntafos as early as 1986~\cite{CN86,CN88,CN91}.
After more than one decade of being in a tangle~\cite{CJN99,Tan01,Tan07},
its polynomial status appears to have been settled by Tan~\etal~\cite{THI99}.
The current fastest algorithm, running in $O(n^4 \log{n})$ time,
is due to Dror~\etal~\cite{DELM03}. A linear-time 2-approximation algorithm is due to Tan~\cite{Tan07}. Ntafos and Gewali~\cite{NG94} showed that a shortest external watchman route for a $n$-vertex convex polygon can be found in  $O(n)$ time. The case of two convex polygons was studied in~\cite{GN98}.
The first polynomial-time approximation algorithm for the watchman route problem in $n$-vertex polygons
with holes was given by Mitchell~\cite{Mi13}; its approximation ratio is $O(\log^2{n})$.
Nilsson and \.Zyli\'nski~\cite{NZ20} showed that computing a shortest tour that sees $k$ specified points in a polygon with $h$ holes is fixed parameter tractable  (FPT) with the parameter $h+k$, but the problem in general cannot be polynomially approximated better than by a factor of $c \log n$ for some constant $c>0$, unless $\P =\NP$.

Regarding the degree of approximation achievable, $\tspn$ for arbitrary
neighborhoods is $\APX$-hard~\cite{BGK+05,SS06}, and approximating $\tspn$
for connected regions in the plane within a factor smaller than 2 is intractable ($\NP$-hard)~\cite{SS06}.
The problem is also $\APX$-hard for disconnected regions~\cite{SS06},
the simplest case being point-pair regions~\cite{DO08}.
It is conjectured that approximating $\tspn$ for disconnected regions in
the plane within a $O(\log^{1/2} n)$ factor is intractable~\cite{SS06}.
Computing the minimum number of point guards, vertex guards, or edge guards,
are all $\APX$-hard~\cite{ESW01} and finding the minimum number of point guards is $\exists \RR$-complete~\cite{AbrahamsenAM22}, even for simple polygons (without holes).
That is, there is a constant $\delta>0$ such that no polynomial-time algorithm achieves an approximation ratio of $1 +\delta$ for any of these problems unless $\P = \NP$.
The survey by Urrutia~\cite{Ur00} gives an introduction to these problems.
See also \cite{BonnetM20} and \cite{Kirkpatrick15} for recent approximations and parameterized hardness results. Historically, watchman routes under limited visibility have been considered by Ntafos~\cite{Nt92};
see also~\cite{Mi00,Mi17} for a survey of the many variants of $\wrp$.
other variants are discussed in~\cite{Mi00,Mi17}.
As mentioned earlier, the problem of computing shortest \emph{external} watchman routes for collections of disjoint polygons was suggested by Ntafos and Gewali~\cite{NG94}.

\paragraph{Definitions and notations.}
A curve is called \emph{simple} if it has no self-intersections.
A \emph{simple polygon} $P$ is a polygon without holes, that is, the interior of the polygon is topologically equivalent
to a disk. A \emph{polygon with holes} is obtained by removing a set of nonoverlapping, strictly interior, simple subpolygons from $P$~\cite{OST17}.
The Euclidean length of a curve $\gamma$ is denoted by $\len(\gamma)$,
or just $|\gamma|$ when there is no danger of confusion.
Similarly, the total (Euclidean) length of the edges
of a geometric graph $G$ or a polygon $P$ is denoted by $\len(G)$ and
$\per(P)$, respectively.

A TSP tour for a set $\F$ of regions (neighborhoods) in $\RR^d$, $d \geq 2$, is a closed curve in the ambient space
that intersects~$\F$ (\ie,  $\gamma$ intersects each region in $\F$).
For $\alpha \geq 1$, an approximation algorithm (for $\orp$, $\ewrp$,  or $\tspn$)
has ratio $\alpha$ if its output tour $\alg$ satisfies $\len(\alg) \leq \alpha \, \len(\opt)$, where $\opt$ is an optimal tour
for the respective problem.

A \emph{convex body} $C \subset \RR^d$ is a compact convex set with nonempty interior.
Its \emph{boundary} is denoted by $\partial C$ and its \emph{interior} by $\mathring{C}$.
The \emph{width} of a convex body $C$ is the minimum width of a strip of parallel lines enclosing $C$.
Informally, a convex body is \emph{fat} if its width is comparable with its diameter.
More precisely, for $0 \leq \lambda \leq 1$, a convex body $C$ is $\lambda$-\emph{fat}
if its width $w$ is at least $\lambda$ times the diameter: $w \geq \lambda \cdot \diam(C)$,
and $C$ is \emph{fat} if the inequality holds for a constant $\lambda$.
For instance, a square is $\frac{1}{\sqrt2}$-fat, a $3 \times 1$ rectangle is $\frac1{\sqrt{10}}$-fat and a segment is $0$-fat.
Let $\gamma$ be a closed curve. The \emph{geometric dilation} of $\gamma$ is
\begin{equation*}
 \delta(\gamma) := \sup\limits_{p,q \in\gamma} \frac{d_\gamma(p,q)}{|pq|},
\end{equation*}
where $d_\gamma(p,q)$ is the shortest distance along $\gamma$ between $p$ and $q$.
For example, the geometric dilation of a the boundary of a square is $2$ and that of a $3 \times 1$ rectangle is $4$.

Points $p$ and $q$ are mutually \emph{visible} if the segment $pq$ does not intersect the interior of any region in $\F$~\cite{OR17}.
An object $O \in \F$ can be \emph{seen} (or \emph{observed}) from a point $p$ if there is a point $q \in \partial O$
such that $p$ and $q$ are mutually visible.
The convex hull of a set $A \subset \RR^d$ is denoted by $\conv(A)$.

\subsection{Our results}

In Section~\ref{sec:prelim}, as a preliminary result, we show that given a set of convex polygons $\F$,
determining whether $\F$ can be observed from a single point can be done by a polynomial time algorithm.
In Section~\ref{sec:approx} we show that given a set of $n$ pairwise disjoint fat convex polygons, an observation tour
whose length is at most $O(\log{n})$ times the optimal can be computed in polynomial time.
The algorithm reduces the $\orp$ problem to $\tspn$ for polygons with holes
and then executes additional local transformations of the tour that only increase the total length by at
most a constant factor.
In Subsection~\ref{subsec:translates}
%In Section~\ref{sec:translates}
we show that the case of translates
(within the same class) allows for a simplification in the algorithm.

\begin{theorem} \label{thm:approx}
Given a family of $n$ pairwise disjoint fat convex polygons, an observation tour
whose length is at most $O(\log{n})$ times the optimal can be computed in polynomial time.
\end{theorem}

In Section~\ref{sec:np-hard}, we prove the $\NP$-hardness of both $\orp$ and $\ewrp$ for both limited and unlimited vision (Theorems~\ref{thm:ORP-is-hard} and~\ref{thm:EWRP-is-hard}).
Throughout this paper, the term \emph{limited vision} refers to unrestricted vision in a given bounding box of the family.

\begin{theorem} \label{thm:ORP-is-hard}
  Given a family of disjoint convex bodies in the plane, computing a shortest observation route is $\NP$-hard.
  (This holds for both limited and unlimited vision.)
  The problem remains so even for families of axis-aligned squares.
\end{theorem}
\begin{theorem} \label{thm:EWRP-is-hard}
  Given a family of disjoint convex bodies in the plane, computing a shortest external watchman route is $\NP$-hard.
   (This holds for both limited and unlimited vision.)
 The problem remains so even for families of axis-aligned squares.
\end{theorem}

In Section~\ref{sec:inapprox}, we prove that one cannot approximate the minimum length of an observation route for $n$ convex bodies in the plane within a factor of $c\,\log n$, for some $c>0$, unless $\P=\NP$.  The inapproximabilty is reduced from \setcov.
\begin{theorem} \label{thm:ORP-inapprox}
  Given a family of $n$ disjoint convex bodies in the plane, the length of a shortest observation route
  cannot be approximated within a factor of $c\,\log n$ unless $\P=\NP$, where $c>0$ is an absolute constant.
  (This holds for both limited and unlimited vision.)
\end{theorem}

In Section~\ref{sec:pentagon} we study the structure of shortest external watchman routes for a convex polygon
(\ie, $|\F|=1$).  In 2006, Absar and Whitesides conjectured that all convex polygons with all angles \emph{obtuse}
have convex-hull routes as their shortest external watchman routes~\cite{AW06}.
Theorem~\ref{thm:pentagon} below refutes this conjecture for every $n \geq 5$.

\begin{theorem} \label{thm:pentagon}
For every $n \geq 5$ there exists a convex polygon with $n$ sides and all angles obtuse such that its
perimeter is \emph{not} a shortest external watchman route.
\end{theorem}

In Section~\ref{sec:comparison} we compare the optimal solutions for the three problems discussed
($\orp$, $\ewrp$,  and $\tspn$).
While the lengths of the optimal tours for these problems can differ substantially for a given input,
we exhibit two natural scenarios when they are roughly the same.
\begin{theorem} \label{thm:roughly-the-same}
  Consider the two scenarios below:
\begin{enumerate}\itemsep -1pt
\item[{\rm (i)}]  Let $\F$ be a maximal packing of unit disks in a large square $S$.
  (That is, one cannot extend this packing by adding new disks contained in $S$). Then
  $\opt_\orp(\F) = \Theta(n)$, $\opt_\ewrp(\F) = \Theta(n)$, and $\opt_\tspn(\F)=\Theta(n)$.
\item[{\rm (ii)}]  Let $\F$ the family of axis-aligned squares in the $\NP$-hardness reduction
 from a set of $n$ integer points in the proof of Theorem~\ref{thm:ORP-is-hard}. Then $\opt_\tspn(\F)| \geq n$,
 $ |\opt_\orp(\F) - \opt_\tspn(\F)| <1$, and $|\opt_\ewrp(\F) - \opt_\tspn(\F)| <1$.
\end{enumerate}
\end{theorem}

\section{Preliminaries}\label{sec:prelim}

Throughout the paper we consider families of disjoint compact convex sets in the plane; and are only concerned
with \emph{external} visibility. See Fig.~\ref{fig:visibility} for an example.

\begin{figure}[ht]
\centering
\includegraphics[scale=0.45]{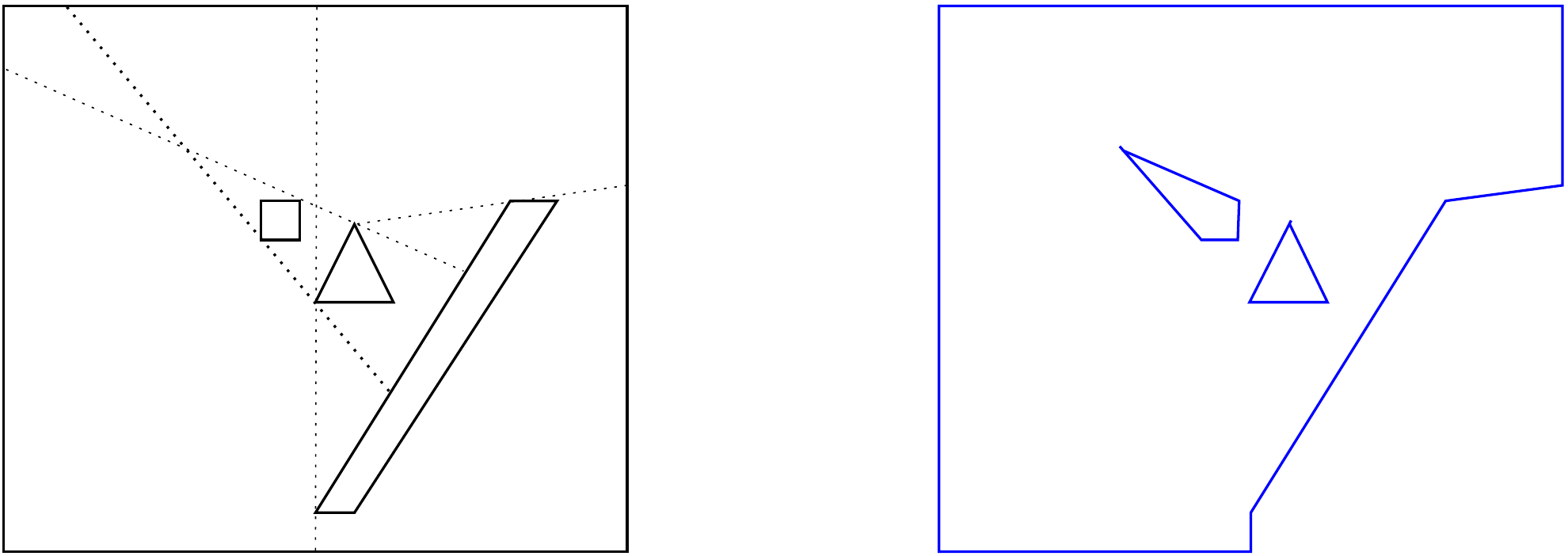}
\caption{Left: three convex polygons in a bounding box.
  Right: the visibility region $V$ of the triangle is a polygon with two holes.}
\label{fig:visibility}
\end{figure}

For a set $\F$ of $n$ disjoint polygons in a rectangle $R$, the \emph{visibility region} of a polygon $P \in \F$, denoted $V(P)$, is the set of all point $p\in R$ such that there exists a point $q\in \partial P$ such that the line segment $pq$ is disjoint from $\overset{\circ}{P'}$ for all $P'\in \F$.

\begin{lemma} \label{lem:connected}
  Given a set $\F$ of $n$ disjoint convex polygons with a total of $m$ vertices in a rectangle $R$, for every $C\in \F$, the visibility region $V(C)$ is a polygon with $O(m+n^2)$ vertices and $O(n^2)$ holes.
\end{lemma}
\begin{proof}
  Note that $C$ is always a hole of $V(C)$.
 The boundary of $V(C)$ is contained in the union of the boundary of the free space $R\setminus \bigcup_{C'\in \F} C'$ and all common inner and outer tangents between $C$ and other polygons in $\F$. Consequently, $V(C)$ is a polygon with $O(m+n^2)$ vertices.
 Since the polygons in $\F$ are disjoint, every hole of $V(C)$ has a vertex at the intersection of two tangents or one tangent and the boundary $\partial C'$, $C'\in \F\setminus \{C\}$. Now $O(n)$ tangents and $O(n)$ convex curves yield $O(n^2)$ intersections, and so the number of holes in $V(C)$ is also $O(n^2)$.

It remains to prove that $V(C)$ is connected. Let $p$ and $q$ be any two points in $V(C)$. By definition, there exist $p',q' \in \partial C$ such that $p'$ is visible from $p$, and $q'$ is visible from $q$.
Let $\rho(p',q')$ denote the shortest path connecting $p'$ and $q'$ on the boundary of $C$.
Then one can reach $q$ from $p$ via the $3$-leg path that connects $p$ to $p'$ via a straight-line segment,
follows $\rho(p',q')$ on $C$'s boundary and connects $q'$ to $q$ via a straight-line segment.
\end{proof}

Next we show that given a set $\F$ of convex polygons, one can determine in polynomial-time whether all polygons in $\F$ can be observed from a single point.
Note that for the variant with unlimited visibility, visibility regions may be unbounded. We need a simple lemma regarding zero-instances of $\tspn$.

\begin{lemma} \label{lem:tspn-zero}
  Given a family $\F$ of (possibly unbounded) polygonal regions with a total of $m$ vertices, one can determine whether there exists a point contained in all polygons in $\F$ %   can be visited by a single point
  (\ie, whether $\opt_\tspn(\F) =0$) in time polynomial in $m$.
\end{lemma}
\begin{proof}
  This is equivalent to determining whether the intersection $\bigcap_{P \in \F} P$ is empty. Since the total complexity of the visibility regions is polynomial in $m$, the complexity of the intersection is also polynomial in $n$. Consequently, the resulting algorithm takes time polynomial in $m$.
\end{proof}

Applying Lemma~\ref{lem:tspn-zero} to the visibility regions of the polygons in $\F$ immediately yields the following.

\begin{corollary} \label{cor:orp-zero}
  Given a set $\F$ of convex polygons with a total of $m$ vertices, one can determine whether $\F$ can be observed from a single point
  (\ie, whether $\opt_\orp(\F) =0$) in time polynomial in $m$.
\end{corollary}

Lemma~\ref{lem:far} below shows that optimizing the length of the route is sometimes impractical as it
may produce routes that are arbitrarily far from the observed objects. However, we can enforce routes in the
near vicinity of the family to be observed by constraining the observation tour to lie in a bounding box of the
family (\eg, an axis-parallel rectangle).

\begin{lemma} \label{lem:far}
  For every $\Delta>0$, there exists a configuration $\F=\F(\Delta)$ of $O(1)$ axis-parallel unit squares such that:
  {\rm (i)}~$\diam(\conv(\F))=O(1)$,
  {\rm (ii)}~$\opt_\orp(\F) =0$, \ie, the configuration can be observed from a single point, and
  {\rm (iii)}~every single observation point is at a distance at least $\Delta$ from $\conv(\F)$.
  Alternatively, $\F$ can be realized from unit disks.
\end{lemma}
\begin{proof}
  We exhibit and analyze a configuration (family $\F$) of six axis-parallel unit squares; refer to Fig.~\ref{fig:6squares}.
  An analogous unit disk configuration (with six elements) can be derived from
  a piece of the hexagonal disk packing by slightly shrinking each disk from its center;
  its analysis is left to the reader.

\begin{figure}[ht]
\centering
\includegraphics[scale=0.5]{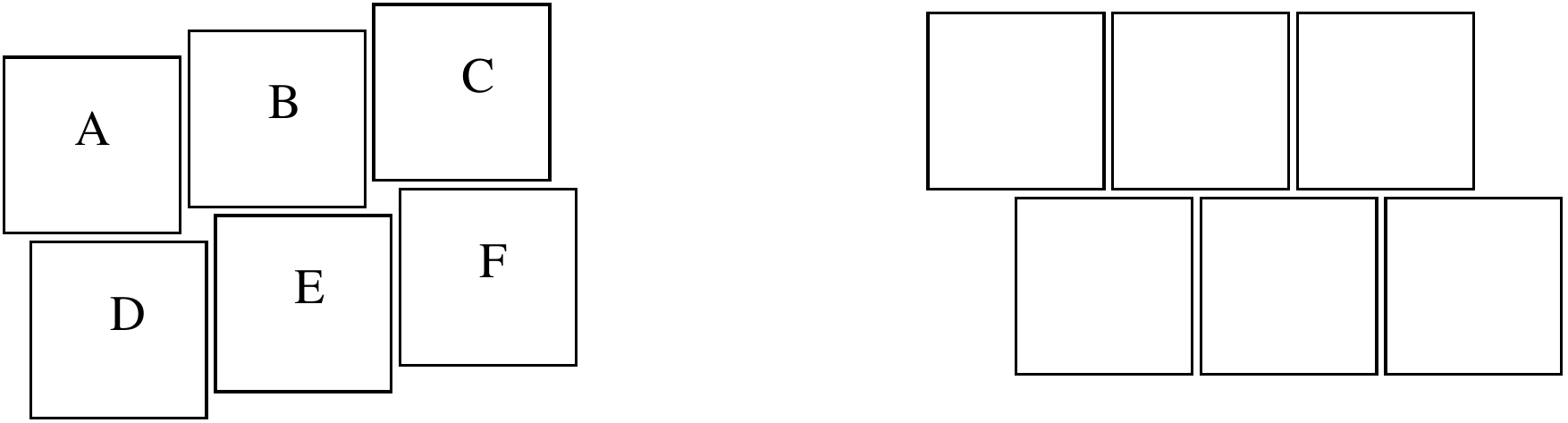}
\caption{Left: This family can be only observed from single points far away up or down.
  Right: This family can be observed from any single point on a horizontal line that separates
  the upper chain of squares from the lower one.}
\label{fig:6squares}
\end{figure}

Let $\eps>0$ be sufficiently small, as specified below.
Successive squares (from left to right) are horizontally separated by $\eps$ and shifted vertically by $2\eps$.
Let $\ell$ be a horizontal line separating the squares $B$ and $E$. Without loss of generality
let $p$ be an observation point on or above $\ell$. If $p$ lies inside $\conv(\F)$, depending on its position,
either $C$ or $D$ is not observable from $p$ (if $\eps>0$ is sufficiently small). Suppose that $p$ lies
outside $\conv(\F)$, and on or to the right of the vertical axis of symmetry of $B$.
Then $D$ is not observable from $p$ unless $p$ is above the common internal tangent to $A$ and $B$ of positive slope
and above the common internal tangent to $B$ and $C$ of negative slope. These tangents continuously depend on $\eps$
and are almost vertical as $\eps$ tends to zero, hence the lowest point in the intersection of the corresponding
halfplanes can be arbitrarily high, as claimed. The case when $p$ lies outside $\conv(\F)$ on or to the left of
the vertical axis of symmetry of $B$ is similar. It is clear that one can always find a suitable $\eps=\eps(\Delta)$,
as required.
\end{proof}

\medskip
The following lemma relates fatness to geometric dilation for closed curves:

\begin{lemma} \label{lem:dilation}
  Let $C$ be a $\lambda$-fat convex curve. Then $\delta(C) \leq \min(\pi \lambda^{-1}, 2(\lambda^{-1} +1))$.
\end{lemma}
\begin{proof}
  If $C$ is a convex curve, it is known~\cite[Lemma~11]{EGK07} that $\delta(C) = \frac{|C|}{2h}$.
  It is also known~\cite[Thm.~8]{DEG+07} that $h \geq w/2$, where $h=h(C)$
  is the \emph{minimum halving distance} of $C$ (\ie, the minimum distance between two points on $C$
  that divide the length of $C$ in two equal parts), and $w=w(C)$ is the width of $C$.
  Putting these together one deduces that $\delta(C) \leq \frac{|C|}{w}$.
  Let $D$ denote the diameter of $C$.
  The isoperimetric inequality $|C| \leq D \pi$ and the obvious inequality $|C| \leq 2D + 2w$ lead to the following
  dilation bounds $\delta(C) \leq \pi \frac{D}{w}$ and $\delta(C) \leq 2\left(\frac{D}{w}+1\right)$,
  see also~\cite{DEG+07,SA00}. Since $C$ is $\lambda$-fat, direct substitution yields the two bounds given in the lemma.
Note that the latter bound is better for small $\lambda$.
\end{proof}

\section{Fat convex polygons}  \label{sec:approx}

In this section we prove Theorem~\ref{thm:approx}. The following algorithm computes a tour for a family $\F=\{C_1,\ldots , C_n\}$ of convex polygons in a rectangle $R$.

\bigskip
\noindent{\bf Algorithm 1.}
\begin{itemize} \itemsep 1pt
\item[] {\sc Step 1:}  Compute the visibility regions $V_i=V(C_i)$, $i=1,\ldots,n$.
\item[] {\sc Step 2:}  Use the $\tspn$ algorithm for connected regions as a subroutine (from~\cite{MM95},\cite{GL99}, or~\cite{EFS09}, as explained in Section~\ref{sec:intro})   to obtain a $O(\log{n})$  approximation for a tour $T$ that visits all $V_i$, $i=1,\ldots,n$. \item [] {\sc Step 3:} Locally transform $T$ by making detours that avoid the elements $C_i \in \F$ that are crossed by   $T$, if any. Specifically, for each $C \in \F$ that is intersected by $T$, replace the subpath $\varrho = C \cap T$ by the shortest path along $\partial C$ connecting the start and end points of $\varrho$; as shown in Fig.~\ref{fig:detour}.  Output the resulting tour $T'$.
\end{itemize}

\begin{figure}[ht]
\centering
\includegraphics[scale=0.8]{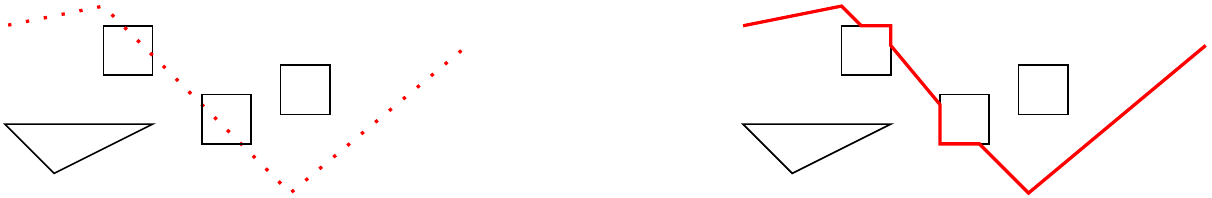}
\caption{Local replacements to obtain $T'$ from $T$.}
\label{fig:detour}
\end{figure}

\old{
\bigskip
\noindent{\bf Algorithm 1.}

 \noindent {\sc Step 1:}  Compute the visibility regions $V_i=V(C_i)$,  for all $i=1,\ldots,n$.

 \noindent {\sc Step 2:}  Use the $\tspn$ algorithm for connected regions as a subroutine
  (from~\cite{MM95}, \cite{GL99}, or~\cite{EFS09}, as explained in Section~\ref{sec:intro})
  to obtain a tour $T$ that visits $V_i$ for all $i=1,\ldots,n$, such that $\len(T)$ is $O(\log{n})$ times the minimum length of such a tour.

 \noindent {\sc Step 3:} Locally transform $T$ by making detours that avoid the elements $C_i \in \F$ that are crossed by $T$, if any. Specifically, for each $C \in \F$ that is intersected by $T$, replace the subpath $\varrho = C \cap T$ by the
  shortest path along $\partial C$ connecting the start and end points of $\varrho$; as shown in Fig.~\ref{fig:detour}.
  Output the resulting tour $T'$.
} % old

\paragraph{Algorithm analysis.}
By Lemma~\ref{lem:connected}, all visibility regions $V_i=V(C_i)$ are connected, and so
$\F$ represents a valid input for the $\tspn$ algorithm.
Recall that the tour $T$ returned by the $\tspn$ algorithm visits all visibility regions; this means that each $C_i$ is seen from some point in $T\setminus(\bigcup_{i=1}^n C_i)$.
The local replacements in {\sc Step~3} ensure that the resulting tour $T'$ does not intersect (the interior of) any obstacle, and maintain the property that the tour visits all visibility regions; \ie, each $C_i$ is seen from some point in $T' \setminus (\bigcup_{i=1}^n C_i)$.
Consequently, $T'$ is an observation route for $\F$.
Since each region $V_i$ has polynomial complexity by Lemma~\ref{lem:connected} (\ie, polynomial in the total number of vertices of the polygons in $\F$),  Algorithm~1 runs in polynomial time.

It remains to bound $\len(T')$ from above. Let $\V=\{V_1,\ldots,V_n\}$.
An observation route for $\F$ must visit the visibility regions $V_i$ for all $i=1,\ldots , n$, and so $\opt_\tspn(\V) \leq \opt_\orp(\F)$, which implies
\[ \len(T)\leq O(\opt_\tspn(\V) \, \log n) \leq O(\opt_\orp(\F)\, \log n). \]
Recall that all $C_i$ are fat and thus by Lemma~\ref{lem:dilation}, the local replacements in {\sc Step~3} increase the length by at most a constant factor (that depends on $\lambda =\Omega(1)$), that is, $\len(T')\leq O(\len(T))$.
This concludes the proof of Theorem~\ref{thm:approx}.

\subsection{Translates of a fat convex polygon}  \label{subsec:translates}

In this subsection we restrict ourselves to families of translates of a fat convex body $C$.
The visibility regions associated to the elements of $\F$ are then restricted and this
allows for a simplified input for the algorithm. Specifically, Corollary~\ref{cor:translate} below shows that each visibility region is a polygon with a unique hole (the convex body itself).  The approximation ratio remains the same.

\begin{lemma} \label{lem:cone}
  Let $C$ be a convex body and let $op$ and $oq$ be the two tangent rays incident to $o$ and $C$ so that $\angle{poq}$
  is oriented counterclockwise. Let $C'$ be a translate  of $C$ further away from $o$ that is tangent to $op$ at $p'$: that is,
  $C'= C + \overrightarrow{pp'}$. Then the ray $oq$ does not intersect~$C'$.
\end{lemma}
\begin{proof}
  Assume for concreteness that $op$ is a horizontal line, and so the points $o,p,p'$ appear on the line in this order.
  Refer to Fig.~\ref{fig:cone}.
\begin{figure}[htbp]
\centering
\includegraphics[scale=0.8]{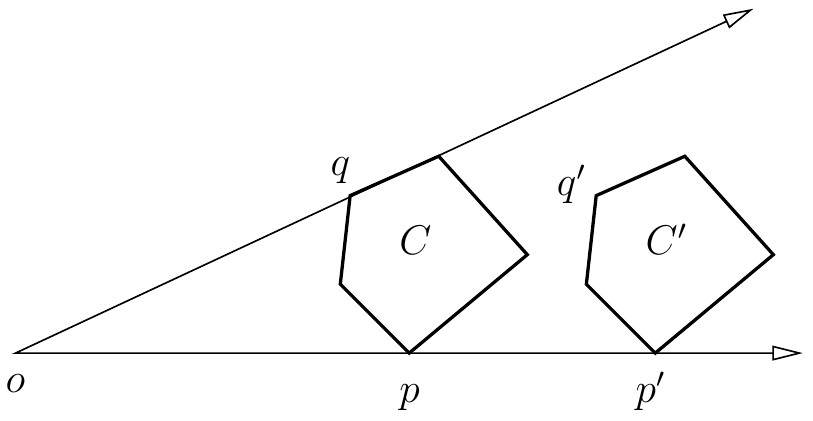}
\caption{Two translates of a convex polygon in a cone.}
\label{fig:cone}
\end{figure}

We may assume that $\angle{poq}$ is acute (since otherwise the claim is obvious).
Recall that $C'= C + \overrightarrow{pp'}$, and consider the following continuous motion:
$C$ moves right and so the contact between $C$ and the ray $oq$ disappears right after the start of the motion.
Thus $C$ remains strictly below this ray throughout the translation and the lemma follows.
\end{proof}

\begin{lemma}\label{lem:ray}
For every $C\in \F$ and every point $o\notin C\cup V(C)$, there exists a ray $\vec{\rho}$ emanating from $o$ such that $\vec{\rho}$ is disjoint from $C\cup V(C)$.
\end{lemma}
 \begin{proof}
Let $op$ and $oq$ be the two tangents to $C$ emanating from $o$ such that the angle $\angle poq$ is oriented counterclockwise. See Fig.~\ref{fig:OneHole}~(right) where the angle is shaded.
By Lemma~\ref{lem:cone}, if a convex body $C'\in \F$, $C' \neq C$, intersects the triangle $\Delta{poq}$, then $C'$ intersects at least one of the line segments $op$ or $oq$. We distinguish between two cases.

\begin{figure}[htbp]
\centering
\includegraphics[scale=0.8]{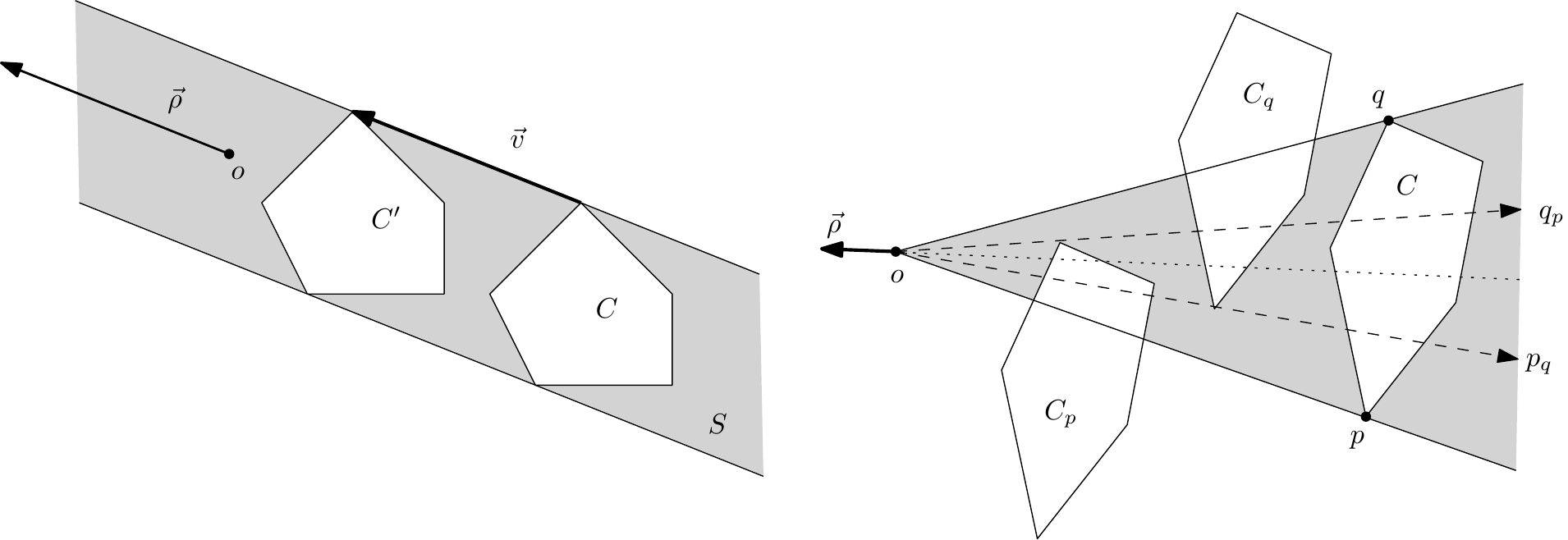}
\caption{Illustrations for Lemma~\ref{lem:ray}: Two possible choices for ray $\vec{\rho}$.}
\label{fig:OneHole}
\end{figure}

\noindent\textbf{Case~1: There exists a convex body $C'\in \F$, $C'\neq C$, that intersects both $op$ and $oq$.} Refer to Fig.~\ref{fig:OneHole}~(left). Since $C'\in \F$, then $C'=C+\vec{v}$ for some  vector $\vec{v}\in \mathbb{R}^2$. Both common external tangents of $C$ and $C'$ are parallel to $\vec{v}$, and so they bound a parallel strip denoted by $S$. Let $\vec{\rho}$ be a ray emanating from $o$ in direction $\vec{v}$. Then $\vec{\rho}$ lies in $S$, and $C'$ separates $C$ from $\vec{\rho}$ within $S$. Consequently, $\vec{\rho}$ is disjoint from $C\cup V(C)$.

\smallskip
\noindent\textbf{Case~2: There is no convex body $C'\in \F$, $C'\neq C$, that intersects both $op$ and $oq$.} See Fig.~\ref{fig:OneHole}~(right). Since $o\notin V(C)$, then $op$ intersects the interior of some body in $\F$. Let $C_p$ be the translate in $\F$ that intersects $op$ and maximizes the angle $\angle poq_p$, where $oq_p$ is a tangent ray to $C_p$.  Similarly, let $C_q$ be the translate in $\F$ that intersects $oq$ and maximizes the angle $\angle p_qoq$, where $op_q$ is a tangent ray to $C_q$.  Since $o\notin V(C)$, and for every $c\in \partial C$, the segment $oc$ intersects the interior of some convex body in $\F$ that also intersects $op$ or $oq$, then the clockwise angle $\angle q_p o p_q$ is nonzero. Let $\vec{u}$ be the direction vector of the angle bisector of $\angle q_p o p_q$, and let $\vec{\rho}$ be a ray emanating from $o$ in direction $-\vec{u}$.

It remains to show that $\vec{\rho}$ is disjoint from $C\cup V(C)$. Assume w.l.o.g.\ that $o$ is the origin and $\vec{\rho}$ is the negative $x$-axis. Let $r$ be any point on $\vec{\rho}$. Consider the horizontal strip $S_p$ bounded by the two horizontal tangent lines to $C_p$. Note that this strip contains the $x$-axis, and $C_p$ separates $\vec{\rho}$ from $C$ within $S_p$. Consequently any line segment between $r$ and a point $c\in \partial C$ below the $x$-axis intersects the interior of $C_p$. Similarly, any line segment between $r$ and a point $c\in \partial C$ above the $x$-axis intersects the interior of $C_q$. Overall, no point in $\vec{\rho}$ can see any point on $\partial C$. This implies that $\vec{\rho}$ is disjoint from $C\cup V(C)$.
\end{proof}

\begin{corollary} \label{cor:translate}
  For every $C \in \F$, the visibility region $V(C)$ is a polygon with exactly one hole, namely~$C$.
\end{corollary}
\begin{proof}
    It is clear that $C\in \F$ is a hole in $V(C)$, since every point in a small neighborhood of $\partial C$ is in $V(C)$, but points in $C$ are not in $V(C)$. Suppose, for contradiction, that $V(C)$ has another hole $H$, i.e., a bounded connected component of $\mathbb{R}^2\setminus V(C)$. Let $p$ be an arbitrary point in the interior of $H$. By Lemma~\ref{lem:ray}, there is a ray $\overrightarrow{r}$ emanating from $p$ that lies entirely in the exterior of $V(C)$, hence in $H$. Consequently, $H$ is unbounded, which is a contradiction.
\end{proof}

See Figure~\ref{fig:translates} for an example with a family of translates.

\begin{figure}[ht]
\centering
\includegraphics[scale=0.6]{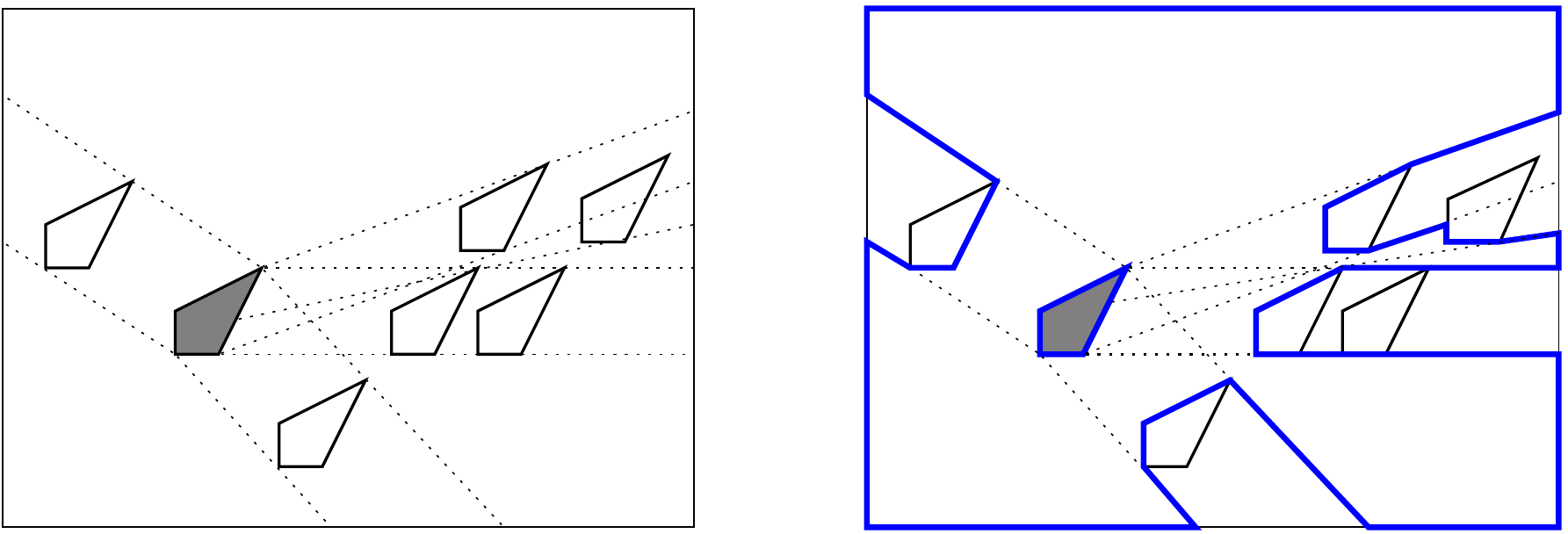}
\caption{Left: a family of seven translates.
  Right: The visibility region of the shaded translate is a polygon with one hole (in blue).}
\label{fig:translates}
\end{figure}

\paragraph{Visibility from a single point.} The algorithm in Corollary~\ref{cor:orp-zero} can be further simplified for families of translates of a convex polygon. Let $V_i=V(C_i)$. Observe that $V_i \cup C_i$ is a simple polygon (without holes).
Since $\partial C_i \subset V_i$, we have $ V_i \cup C_i = V_i \cup \mathring{C_i}$.
Moreover, we have $V_i \cap \mathring{C_j} =\emptyset$ for every $i \neq j$. Consequently, by using
the distributivity of intersection over union we obtain
\begin{align*}
  (V_i \cup C_i) \cap (V_j \cup C_j) &=  (V_i \cup \mathring{C_i}) \cap (V_j \cup \mathring{C_j}) \\
  &= (V_i \cap V_j) \cup (V_i \cap \mathring{C_j}) \cup (V_j \cap \mathring{C_i})
  \cup  (\mathring{C_i} \cap \mathring{C_j}) \\
 &= (V_i \cap V_j) \cup \emptyset \cup \emptyset \cup \emptyset \\
  &= V_i \cap V_j.
\end{align*}
Consequently, it suffices to work with simple polygons. Indeed,
\[ \bigcap_{i=1,\ldots,n} V_i = \bigcap_{i=1,\ldots,n} (V_i \cup C_i), \]
and so it suffices to determine whether the intersection of the simple polygons $(V_i \cup C_i)$,
$i=1,\ldots,n$,  is non-empty.

\section{$\NP$-hardness results}  \label{sec:np-hard}

In this section we prove that both $\orp$ and $\ewrp$ are  $\NP$-hard.
Both reductions are from Rectilinear $\tsp$ ($\rtsp$) assuming limited vision.
Essentially the same reductions also hold with unlimited vision.
It is known~\cite{GGJ76,P77} that both $\etsp$ and $\rtsp$ are $\NP$-hard.
The three relevant problems
$\orp$, $\etsp$, and $\rtsp$ (see~\cite{GJ79,P77}) can be formulated as decision problems
as follows:

\begin{quote}
$\orp$: Given a family $\F$ of $k$ polygonal regions and a positive integer $m$,
does there exist an observation tour of Euclidean length at most $m$?
\end{quote}

\begin{quote}
$\etsp$: Given a set of $n$ points in the plane and a positive integer $m$,
does there exist a tour of Euclidean length at most $m$ that visits all the points?
\end{quote}

\begin{quote}
$\rtsp$: Given a set of $n$ points in the plane and a positive integer $m$,
does there exist a tour of rectilinear length at most $m$ that visits all the points?
\end{quote}

\begin{figure}[htb]
\begin{center}
\includegraphics[width=0.97\columnwidth]{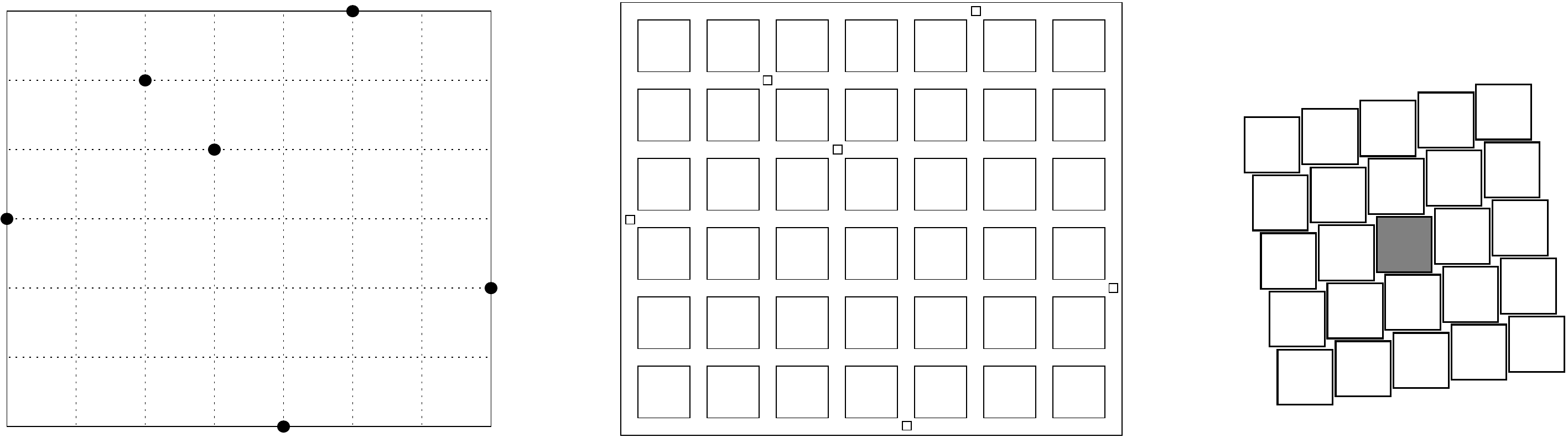}
\caption{Left: point set $S$ in $R$. Middle: family of axis-aligned squares in $B$.
  Right: a small cluster (gadget) $\G$ of $25$ disjoint congruent squares for each point in $S$.
  The shaded square in the center of $\G$ cannot be observed (or watched) without
  entering the convex hull $\conv(\G)$.
}
\label{fig:reduction}
\end{center}
\end{figure}

It is known from that $\etsp$ is $\NP$-hard under both $L_1$- and $L_2$-norms~\cite{GGJ76,P77}.
The $\NP$-hardness of $\wrp$ was first announced by Chin and Ntafos~\cite{CN86,CN88} via a reduction from $\etsp$ under the $L_2$-norm. However, the first valid proof,
due to Dumitrescu and T\'oth~\cite{DT12}, makes a reduction from $\rtsp$. To show that $\orp$ is $\NP$-hard, we use a reduction based on a similar idea, as follows;
refer to Fig.~\ref{fig:reduction}.

\paragraph{Proof of Theorem~\ref{thm:ORP-is-hard}.}
Given a set $S$ of $n$ points of the integer lattice $\mathbb{Z}^2$, assume without loss of generality that the smallest axis-aligned rectangle containing $S$ is $R=[0,a] \times [0,b]$,
for some positive integers $a \geq b$; see Fig.~\ref{fig:reduction}.
Construct a family $\F$ of $ab+25n$ axis-parallel squares contained in a rectangle $B \supset R$ as follows.
There are two types of squares: (i)~large squares that correspond to the cells of the grid, and
(ii)~small squares grouped in clusters of $25$, where each cluster corresponds to a point in $S$.
More precisely,  we have $ab$ large squares formed by the cells of the grid, but only slightly smaller so that they
are disjoint. We also have $n$ small clusters of $25$ small squares each.
For each cluster, the middle square can only be seen by entering the convex hull of the cluster.
The clusters are small enough so that they fit in the narrow corridors left by the big squares.
The center of the central square in a cluster is the \emph{reference} point of the cluster;
for each point in $p \in S$, the corresponding cluster $C(p)$ has its reference point at $p$.

The width of the narrow corridors formed by the large squares is set to $w=1/(10 a n)$.
The side-length of each small square is set to $s=\frac{w}{100}$.
We set $B=[-\frac{w}{2},a+\frac{w}{2}] \times [-\frac{w}{2},b+\frac{w}{2}]$, where the $ab$ large squares are
defined as $[i+\frac{w}{2},i+1-\frac{w}{2}]\times [j+\frac{w}{2},j+1-\frac{w}{2}]$
for $0 \leq i \leq a-1$ and $0 \leq j \leq b-1$. Between any two adjacent large squares,
there is a narrow rectangular corridor of length $1-w$ and width $w$. There are also narrow corridors
of the same dimensions between the boundary of $B$ and the adjacent large holes.

The reduction, hence the $\NP$-hardness, follows via the following claim.

\medskip
\noindent
\emph{Claim.} For a positive integer $m$, there exists a tour of $S$ of rectilinear length $m$ if and only if
there exists an observation tour of $\F$ of Euclidean length $m+\delta$, with $-0.4 \leq \delta \leq 0.4$.

\medskip
To verify the claim observe first that the rectilinear distance between
any two (lattice) points in $S$ is an integer. Hence the total rectilinear length
of the shortest tour of the points in $S$ is an integer, say $m$.
Observe also that any tour of the points can be converted into a tour of the family of squares,
and vice versa, by visiting the points in $S$ and the corresponding clusters in the same order.
Moreover, we show below that the lengths of the two tours are very close to each other.
Indeed, on the one hand, by making only small detours from any given $\tsp$ tour for $S$
yields an $\orp$ tour of $\F$. On the other hand, an $\orp$ tour of $\F$ can be
converted into a $\tsp$ tour for $S$ whose length is very close to the original one.
Note the following two properties of sub-paths that visit two consecutive clusters :
\begin{enumerate}\itemsep 0pt
\item If the $L_1$-distance between two points in $S$ is an integer $d$, $1 \leq d \leq a+b$,
  then any path between two points from which the central squares of the corresponding clusters
are visible has length at least $d(1-w)$, since any path has to traverse at least $d$ narrow corridors.
\item If the $L_1$-distance between two points in $S$ is an integer $d$, then there is a path of length at most
  $d+w \leq d(1+w)$ between any two points in the interiors of the corresponding clusters,
  since it takes a detour no longer than $w$ to observe all the squares in a cluster.
  At most $n$ such small detours are needed, and each adds at most $w$ to the total length.
\end{enumerate}
It follows that the rectilinear length of the shortest tour of the $n$ points in $S$ can differ from the (Euclidean) length
of the shortest external watchman tour of $\F$ by at most $2n(a+b)w = 2n(a+b)/(10 a n) \leq 0.4$, as required.
\qed

\paragraph{Proof of Theorem~\ref{thm:EWRP-is-hard}.}
We proceed with the same reduction and a similar claim.

\smallskip
\noindent
{\em Claim.} For a positive integer $m$, there exists a tour of $S$ of rectilinear length $m$ if and only if
there exists an external watchman tour of $\F$ of Euclidean length $m+\delta$, with $-0.4 \leq \delta \leq 0.4$.

\medskip
It suffices to notice that
(i)~since $R$ is the smallest axis-aligned rectangle containing $S$, visiting all the clusters
will automatically guarantee seeing the entire boundary for each of the large squares,
and (ii)~it takes a detour no longer than $w$ to see the entire boundary for each square in a visited cluster.
\qed

\paragraph{Remarks.}
Observe that the integrality requirement for $m$ is crucial.
Furthermore, as in the reduction from~\cite{DT12}, no such claim holds if the length of the tour of the points in $S$
is measured in the $L_2$-norm. Observe also that squares of only two different sizes are used in the reduction.
We conjecture that both $\orp$ and $\ewrp$ remain $\NP$-hard even for axis-aligned unit squares.

\section{Inapproximability results}
\label{sec:inapprox}

We deduce the inapproximability of \orp\ from that of \setcov. A \emph{set system} is a pair
$(U, \mathcal{S})$, where $U$ is a set and $\S$ is a collection of subsets of $U$. Given a set system $(U,\mathcal{S})$, the \setcov\ problem asks for the minimum number of sets in $\mathcal{S}$ whose union is $U$. \setcov\ cannot be approximated within a factor of $(1-o(1)) \ln n$ unless $\P=\NP$~\cite{DinurS14}, where $n=|U|$. Furthermore, for any $c\in (0,1)$, \setcov\ cannot be approximated within a factor of $c\, \ln n$ over instances where $m\leq O(n^{f(c)})$ for some function $f:(0,1)\to \RR$,
% and $\opt_{SC}\geq \Omega(\log n)$,
unless $\P=\NP$~\cite{DinurS14,Nelson07}; see also \cite{Feige98,LundY94,Moshkovitz15,Nelson07}.

Given a set system $(U,\mathcal{S})$ with $|U|=n$ and $|\mathcal{S}|=m$, we construct a family $\F$ of disjoint convex polygons in four stages. We first construct an arrangement of lines $\mathcal{L}$ in $\RR^2$, and then ``thicken'' the lines into narrow corridors. The family $\F$ will consist of the convex faces of this arrangement and $n+3$ additional axis-parallel rectangles inserted in the corridors at strategic locations. We continue with the details; see
Fig.~\ref{fig:inapprox}.

\begin{figure}[htb]
\begin{center}
\includegraphics[width=\columnwidth]{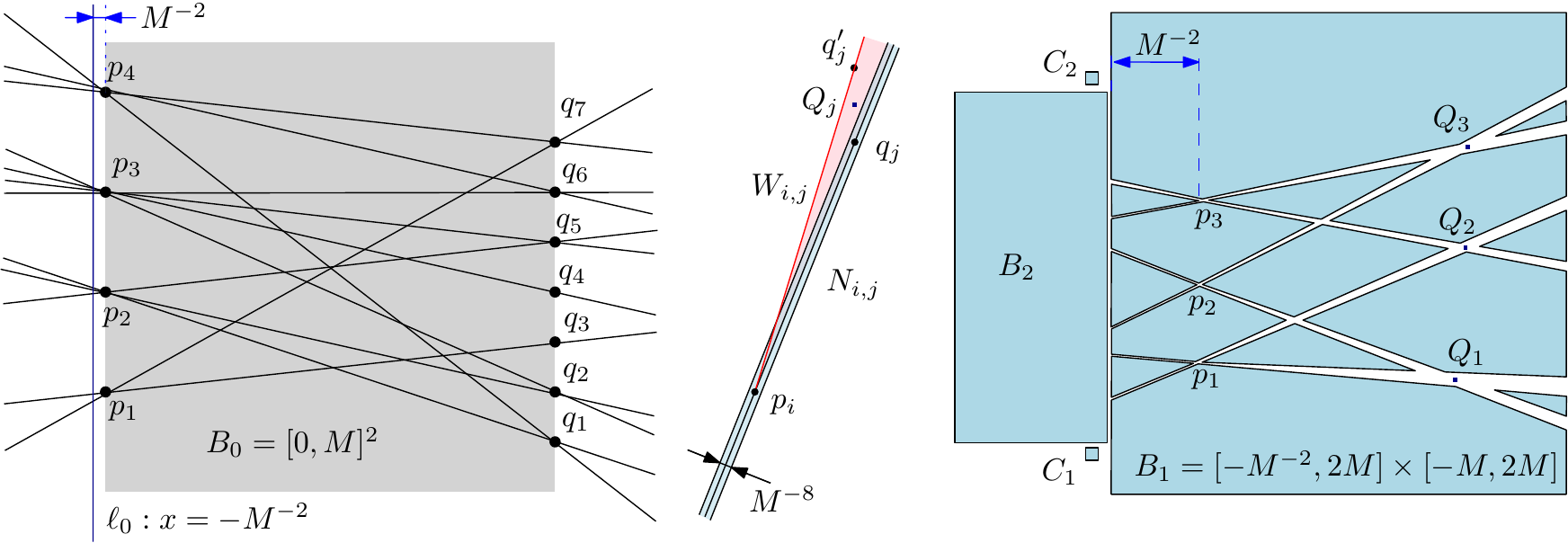}
\caption{Left: A line arrangement $\mathcal{L}$ constructed in Stage~1; here $n=7$ and $m=4$.
  Middle: A slab $N_{i,j}$ and a wedge $W_{i,j}$ for a line $p_iq_j\in \mathcal{L}$.
  Right: A schematic picture of the family $\F$ (not to scale).
}
\label{fig:inapprox}
\end{center}
\end{figure}

\smallskip
\noindent \textbf{Stage~1.}
We are given a set system $(U,\mathcal{S})$ with $U=\{1,\ldots ,n\}$ and $\mathcal{S}=\{S_1,\ldots , S_m\}$. Let $M=\max\{12m,n+1\}$, and consider the axis-aligned square $B_0=[0,M]^2$. Each set $S_i$, $i\in \{1,\ldots ,m\}$, is represented by the point $p_i=(0,i)$ at the left side of $B_0$, and each element $j\in U$ corresponds to the point $q_j=(M,j)$ at the right side of $B_0$. Let $\mathcal{L}$ be the set of lines $p_iq_j$ such that $j\in S_i$.

\smallskip
\noindent \textbf{Stage~2.} For each line $p_iq_j\in \mathcal{L}$, let $N_{i,j}$ be the $\frac12\, M^{-8}$-neighborhood of $\ell$, which is a slab of width $M^{-8}$. For each line $p_iq_j\in \mathcal{L}$, we also create a cone $W_{i,j}$ bounded by the ray $\overrightarrow{p_iq_j}$ from below and the ray $\overrightarrow{p_i q_j'}$ from above, where $q_j'=(M, j+M^{-4})$ is a point at distance $M^{-4}$ above $q_j$, and in particular $q_j'\notin N_{i,j}$. Note that $N_{i,j}\cup W_{i,j}$ is a simply connected region; see Fig.~\ref{fig:inapprox}(middle).

Consider the axis-aligned rectangle $B_1=[-M^{-2},2M]\times [-M,2M] \supset B_0$.
Let $\F_1$ be the bounded components of $B_1\setminus \bigcup_{i,j} (N_{i,j}\cup W_{i,j})$,
that is, we create a convex polygon for each bounded face of the arrangement $\mathcal{L}$ within $B_1$.

\smallskip
\noindent \textbf{Stage~3.} Create a family $\F_2$ of the following $n+3$ disjoint axis-aligned rectangles: a small square $Q_j$ of side length $M^{-8}$ centered at the midpoint of $q_jq_j'$ for all $j\in U$; two unit squares, denoted $C_1$ and $C_2$, resp., centered at $(-1,-1)$ and $(-1,M+1)$; and a large rectangle $B_2=[-M,-M^{-2}-M^{-4}]\times [0,M]$.

\smallskip
\noindent \textbf{Stage~4.} Apply the linear transformation $T:\RR^2\to \RR^2$, $T(x,y)=\left(\frac{M^2 x}{2},\frac{M^{-1} y}{24}\right)$,
and let $\F=\{T(C): C\in \F_1\cup \F_2\}$.

\begin{lemma}\label{lem:inapprox}
For an instance $(U,\mathcal{S})$ of \setcov~ with $|U|=n$ and $|\mathcal{S}|=m$, let $\mathcal{F}$ be the family of disjoint convex polygons constructed above.
\begin{enumerate}
\item There is a polynomial $f(m,n)$ such that the total number of vertices of the polygons in $\F$ is at most $f(m,n)$,  each vertex has rational coordinates where both numerators and denominators are bounded by $f(m,n)$.
\item For every integer $k$, $1\leq k\leq m$, the union of $k$ sets in $\mathcal{S}$ covers $U$ if and only if $\mathcal{F}$ admits an observation tour of length at most $k+\frac12$.
\end{enumerate}
\end{lemma}
\begin{proof}
(1) There are at most $mn$ lines in the arrangement $\mathcal{L}$ in Stage~1, each line can be written in the form $ax+by=c$ with integer coefficients in the range $[-M,M]$. Each line $p_iq_j\in \mathcal{L}$ corresponds to a narrow corridor $N_{i,j}\cup W_{i,j}$ bounded by three lines. The vertices of the polygons in $\F_1$ are intersection points of boundaries of such corridors: $O(mn)$ lines yield $O(m^2n^2)$ intersection points, which all have rational coordinates bounded by a polynomial in $m$ and $n$. The $n+3$ rectangles in $\F_2$ are defined explicitly, and they also have rational coordinates bounded by polynomials in $m$ and $n$. Finally, the linear transformation in Stage~4 maintains these properties.

\medskip\noindent
(2a) Assume that $k$ sets in $\mathcal{S}$ jointly cover $U$, that is, $\bigcup_{t=1}^k S_{i_t}=U$ for some $i_1,\ldots , i_k\in \{1,\ldots m\}$. We construct an observation tour for $\mathcal{F}$. We describe the tour in terms of the polygons before Stage~4, since a linear transformation maintains visibility (but distorts distances). Let the initial tour $\gamma_0$ traverse the left side of the rectangle $B_1$ twice.
The upper-left and lower-left corners of $B_1$ see the squares $C_1$ and $C_2$, and every point in $\gamma_0$ can see $B_2$. The tour $\gamma_0$ intersects every line $p_iq_j\in \mathcal{L}$. The point $\gamma_0\cap p_iq_j$ can see all convex bodies in $\F_1$ whose boundaries touch the slab $N_{i,j}$. Since the upper arc of every polygon in $\F_1$, other than the polygon containing the top side of $B_1$, is formed by the bottom sides of the slabs $N_{i,j}$, and so $\gamma_0$ can see every polygon in $\F_1$.

We expand $\gamma_0$ to a tour of $\F_1 \cup \F_2$ as follows: For $t=1,\ldots ,k$, choose an arbitrary point $q_{j_t}\in S_{i_t}$, and add a loop from the point $p_{i_t}q_{j_t}\cap \gamma_0$ to the point $p_{i_t}$ in the corridor $N(p_{i_t}q_{j_t})$. The point $p_{i_t}$ can see the small squares representing all $j\in S_{i_t}$. Since $\bigcup_{t=1}^k S_{i_t}=U$, every small square in $\F_2$ is visible from some point in the tour. After the linear transformation in Stage~4, we obtain an observation tour $\gamma$ for $\F$. We bound the length of $\gamma$ using the $L_1$-norm of its edges.
The linear transformation in Stage~4 compresses the $y$-extents of each edge, so the $L_1$ norm is dominated by the $x$-extents: The $x$-extent of an edge between the left side of $B_1$ and $p_i$ is exactly $M^{-2}$, thus the sum of $x$-extents is $2kM^{-2}$. Accounting for the $y$-extents of these $2k$ edges and $\gamma_0$, and applying the linear transformation in Stage~4, we obtain $|\gamma|\leq k+\frac12$, as required.

\medskip\noindent
(2b) Now assume that $\mathcal{F}$ admits an observation tour $\gamma$ with $|\gamma|\leq k+\frac12$. We analyze the construction before Stage~4, hence the sum of $x$-extents of all edges of the tour is at most $(2k+1)M^{-2}$. The observation tour intersects the visibility region $V(Q_j)$ for all $j\in U$.
A square $Q_j$ lies in the corridor $N_{i,j}\cup W_{i,j}$ iff $j\in S_i$, and $Q_j$ is only visible from such corridors. Each line $p_iq_j\in \mathcal{L}$ is incident to two points on opposite vertical sides of a square, and so its slope is in the range $[-1,1]$. By construction, $Q_j\subset W_{i,j}$ and the cone $W_{i,j}$ is convex, consequently $W_{i,j}\subset V(Q_j)$. However, $Q_j$ is above the slab $N_{i,j}$, every vertical segment between them has length at least $\frac13\, M^{-4}$, and so the $x$-coordinate of the leftmost point in $V(Q_j)\cap N_{i,j}$ is at least $-M\cdot M^{-8}/(\frac13\, M^{-4}) = -3\,M^{-3}$. Overall, the $x$-coordinate of the leftmost point in $V(Q_j)$ is at least $-3\,M^{-3}$.

Note that the visibility region $V(C_1)$ is disjoint from $V(Q_j)$ for all $j\in U$. Since $|\gamma|\leq (2k+1)M^{-2}<M^{-1}$, $\gamma$ must be contained in the vertical slab $\{(x,y)\in \RR^2: |x|<M^{-1}\}$.
Consequently, $\gamma$ visits each visibility region $V(Q_j)$ in the vertical slab $\{(x,y)\in \RR^2: -3\,M^{-3}<x<M^{-1}\}$.

Let $I$ be the set of indices $i\in \{1,\ldots , m\}$ such that $\gamma$ visits a
point $g_i\in (N_{i,j}\cup W_{i,j})\cap V(Q_j)$ for some $j\in U$.
Since $\gamma$ is an observation tour, we have $U=\cup_{i\in I} S_i$.
It remains to show that $|I|\leq k$.

The tour $\gamma$ determines a cyclic order on the points $G=\{g_i: i\in I\}$. We claim that the $x$-extent of the arc of $\gamma$ between two consecutive points in $G$ is at least $2M^{-2}-6M^{-3}$.
To prove the claim, note that every point $p_i=(0,2i)$ has integer coordinates,
and the slope of every line $p_iq_j\in \mathcal{L}$ is in the range $[-1,1]$. This implies that if two lines in $\mathcal{L}$ cross in the slab
$\{(x,y)\in \RR^2: |x|<M^{-1}\}$, then they cross at $p_i$ for some $i\in U$.
Consequently, if visibility corridors, say $N_{i,j}\cup W_{i,j}$ and $N_{i',j'}\cup W_{i',j'}$ intersect, then the intersection is in a $M^{-8}$-neighborhood of $p_i$ for some $i\in U$.
It follows that the arc of $\gamma$ between any two points in $G$ must reach the line $x=-M^{-2}$ on the left or the line $x=1$ on the right. In both cases, its arclength is at least $2M^{-2}-6M^{-3}$, as claimed.
We can bound the length of $\gamma$ by the summation of the arcs between consecutive points in $G$:
\[
|\gamma|\geq  |I| (2M^{-2}-6M^{-3})
\geq (2\, |I|) M^{-2} - 6m M^{-3}
\geq \left(2\, |I| -\frac12\right) M^{-2}.
\]
In combination with $|\gamma|\leq (2k+1)M^{-2}$ and $M\geq 12m$, we obtain that $|I|\leq k+\frac34$. As both $k$ and $|I|$ are integers, $|I|\leq k$ follows.
\end{proof}

\paragraph{Proof of Theorem~\ref{thm:ORP-inapprox}.}
As noted above~\cite{DinurS14}, there exists a constant $\kappa>0$ such that \setcov \, cannot be approximated
within a factor of $\kappa \log n$ on instances $(U,\mathcal{S})$ with $n=|U|$, $m=|\mathcal{S}|$, and $m\leq O(n^\alpha)$ for a constant $\alpha>0$ unless $\P=\NP$. By Lemma~\ref{lem:inapprox}, for every such instance $(U,\mathcal{S})$ of \setcov, there is a family $\mathcal{F}$ of $N$ disjoint convex polygons in the plane such that (i) for every $k\leq n$, there is a set cover of size $k$ iff $\mathcal{F}$ admits a observation route of length at most $k+\frac12$, and (ii) $N\leq {mn \choose 2} + n+3 \leq m^2 n^2 = O(n^{2\alpha +2})$.

Suppose that for $\delta=\kappa (2\alpha + 2)^{-1}>0$, there exists a polynomial-time $(\delta\log N)$-factor approximation algorithm for \orp \, with $N$ convex bodies. Since $\delta\log N \leq \delta (2\alpha+2) \log n = \kappa \log n$, this yields a
$(\kappa \log n)$-approximation for \setcov, which is a contradiction unless $\P=\NP$.
\qed

%\paragraph{Remarks.} The reductions above work equally under limited and unlimited vision. To complete the proof of Theorem~\ref{thm:ORP-inapprox}, we extend the reduction to \orp\ with axis-aligned squares. Specifically, consider the family $\mathcal{F}$ of disjoint convex bodies constructed for an instance $(U,\mathcal{S})$ of \setcov. Replace each polygon $P\in T(\mathcal{F}_1)$ with an arrangement of axis-aligned squares in a grid pattern (as in Fig.~\ref{fig:reduction}(middle)). The grids for distinct polygons in $T(\mathcal{F}_1)$ are aligned such that each grid square is visible from the tour $\gamma_0$. However, the squares $Q_1,\ldots , Q_n$ representing $\mathcal{U}$ are sufficiently small such that they are invisible from $\gamma_0$, and each square $Q_j$ is visible only from the cones $W_{i,j}$. Finally, $C_1$, $C_2$, and $B_2$ can easly be replaced by a suitable axis-aligned square with the same functionality.

\section{External watchman tours for a convex polygon}  \label{sec:pentagon}

Given a polygon, the External Watchman Route problem ({\textsc EWRP}) is that of finding a shortest route such that each point in the exterior of the polygon is visible from some point along the route. For a convex polygon, this requirement is tantamount to requiring that each point on the boundary of the polygon is visible from some point along the route.

Let $P$ be a convex polygon. Ntafos and Gewali~\cite{NG94} distinguished between two types of external watchman tours: those that wrap around the perimeter and those that do not. They also showed that the second type of route can be obtained by doubling a simple open curve that wraps around a part of $P$'s boundary and is extended at both ends until the vision
encompasses the entire boundary of $P$; see Fig.~\ref{fig:types} for an example. They referred to the second type as ``a $2$-leg watchman route'', see~\cite{NG94}.

\begin{figure}[ht]
\centering
\includegraphics[scale=0.5]{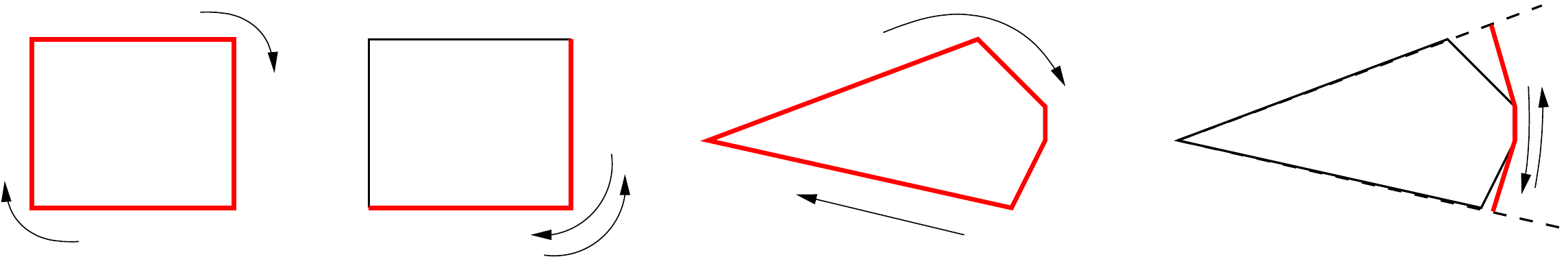}
\caption{The two types of external watchman tours.}
\label{fig:types}
\end{figure}

\paragraph{Proof of Theorem~\ref{thm:pentagon}.}
  Let $P$ be a convex pentagon whose angles listed in clockwise order from the top are
  $120^\circ$, $120^\circ$, $90^\circ$, $90^\circ$, $120^\circ$; refer to Fig~\ref{fig:pentagon}.
  Its horizontal base has length $2$ and each of its vertical sides is of length $\eps$, for a small $\eps>0$.
  Note that the length of each slanted top side is  $a=2/\sqrt3$.
Observe that all angles are at least $90^\circ$. By slightly shortening the base to $2(1-\eps^2)$
the pentagon becomes one with all angles obtuse, say $P'$, whose angles are
$120^\circ$, $(120 -\delta)^\circ$, $(90+\delta)^\circ$, $(90+\delta)^\circ$, $(120-\delta)^\circ$,
for some small $\delta>0$. Its side lengths (listed in the same order) are
$a$, $\eps (1+\eps^2)^{1/2}$, $2(1-\eps^2)$, $\eps (1+\eps^2)^{1/2}$, $a$.
Thus
\[ \per(P') = 2(a + (1 -\eps^2) + \eps (1+\eps^2)^{1/2}) > 2 (1 + a -\eps). \]

\begin{figure}[ht]
\centering
\includegraphics[scale=0.55]{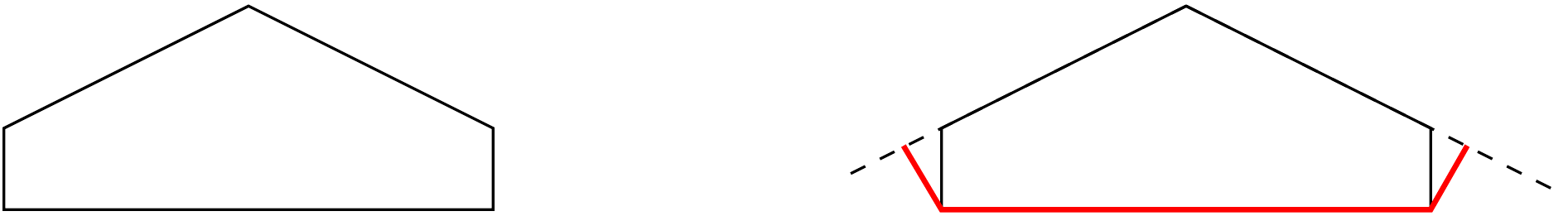}
\caption{Left: a pentagon with all angles obtuse can be obtained by slightly perturbing this one.
  Right: a second type of watchman route can be obtained by doubling the red $3$-chain polygonal curve;
  its two legs are perpendicular to the extensions of the two top sides. The figure is not to scale.}
\label{fig:pentagon}
\end{figure}

Observe that the red curve in the figure is a $3$-polygonal chain that makes a valid watchman path for $P'$.
Its length is $<2(1+\eps)$ and thus doubling it yields an external watchman route of length $<4(1+\eps)$.
It remains to show that $4(1+\eps) < 2 (1 + a -\eps)$ or $\eps < (a-1)/3 = \frac{2 -\sqrt3}{3 \sqrt3}$, which clearly holds for small $\eps>0$. Note that the ratio between the length of the double red curve and the perimeter is in fact smaller than some absolute constant $<1$, \eg, $0.93<1$ in the range $\eps \leq 0.001$.

\medskip
Alternatively, by cutting two small right-angled isosceles triangles,
one from each side of the base of the pentagon one gets a convex heptagon
all whose angles are at least $120^\circ$. A~similar calculation as above shows that
this heptagon provides yet another counterexample in which all angles are even larger, here at least $120^\circ$.

By slightly 'shaving' the pentagon in a repeated manner, one can increase the number
of vertices while maintaining all angles obtuse and thereby obtain counterexamples to Conjecture 1 in~\cite{AW06}
for every $n \geq 5$.
Note, however, that since every triangle or convex quadrilateral has at least one nonobtuse angle, the range for $n$ in Theorem~\ref{thm:pentagon} cannot be improved.
\qed

\section{Problem comparison}
\label{sec:comparison}

In this section we compare the optimal solutions to the three problems  ($\orp$, $\ewrp$, and $\tspn$)
on various instances. While such solutions can differ substantially, they can be also
close to each other in certain natural scenarios (as in Theorem~\ref{thm:roughly-the-same}).
As a result, we get a better understanding of these problems.

\begin{observation} \label{obs:simple}
  Let $\F_1 \subseteq \F_2$ be two families of convex bodies. Then
  \begin{enumerate} \itemsep 1pt
    \item $\opt_\tspn(\F_1) \leq \opt_\tspn(\F_2)$
    \item $\opt_\orp(\F_1) \leq \opt_\orp(\F_2)$
    \item $\opt_\ewrp(\F_1) \leq \opt_\ewrp(\F_2)$
    \end{enumerate}
\end{observation}

We start by observing a key difference between the range of optimal solutions for these problems.
Specifically, we have the following lower bound on the length of an external watchman route.

\begin{lemma} \label{lem:non-zero}
  Let $C$ be a convex body and $W$ be an external watchman route for $C$. Then $\len(W) >0$.
  Similarly, let $\F$ be a family of disjoint convex bodies. Then $\opt_\ewrp(\F)>0$.
\end{lemma}
\begin{proof}
  Assume that the boundary $\partial C$ is visible from a single point $o$. Take
  a ray from $o$ that goes through the interior of $C$ and intersects $\partial C$ at $p$ and $q$, respectively.
  Then $q$ is not visible from $o$, a contradiction, and the first part in the lemma follows.
  For the second part, let $C$ denote an arbitrary element of $\F$. By Observation~\ref{obs:simple},
$\opt_\ewrp(\F) \geq \opt_\ewrp(\{C\})>0$, as claimed.
\end{proof}

We next confirm the intuition ``observing is easy, traveling is expensive''.

\begin{lemma} \label{lem:compare1}
  There exist families $\F$ of $n$ congruent disks for which %\linebreak
  $\opt_\ewrp(\F) \ll \opt_\tspn(\F)$.
 In particular $\opt_\ewrp(\F) = \Theta(1)$ and %\linebreak
 $\opt_\tspn(\F)=\Theta(\sqrt{n})$.
\end{lemma}
\begin{proof}
Assume that $n$ is a perfect square.
Place a disk of radius $\eps>0$ centered at each point of the integer lattice $[1,\sqrt{n}]^2$.
Scale down the construction so that it fits inside the unit square $[0,1]^2$.
  If the radius is sufficiently small, there  is a watchman route for $\F$ of length $\Theta(1)$ that goes around
  and close to the perimeter of $\conv(\F)$; for a similar setup recall the classical {\em orchard visibility problem}
  due to P\'olya~\cite{Pol1918}.  On the other hand, the length of the shortest tour visiting each disk is
  clearly $\Theta(\sqrt{n})$ (recall also Few's result~\cite{Fe55}).
\end{proof}

On the other hand, when disks are densely packed, one expects that the length of an optimal $\tspn$ tour
approximates well the length of an optimal watchman tour. Indeed, the intuition is that one cannot see too far
in a densely packed forest of congruent trees and in order to see each tree one essentially needs to visit
the entire forest. Moreover, we show that in this scenario the length of a shortest $\orp$ tour,
$\ewrp$ tour, and $\tspn$ tour are roughly the same, namely within constant factors from each other.
We recall the following result of Dumitrescu and Jiang~\cite{DJ11a}
confirming a conjecture by Mitchell. (It is likely that the result holds for a much smaller radius threshold, possibly $<100$.)

\begin{theorem} \label{thm:hide} {\rm \cite{DJ11a}}
Any dense (circular) forest with congruent trees (of unit diameter) that is deep enough has a hidden point.
\end{theorem}

Using Theorem~\ref{thm:hide} on one hand and Few's technique~\cite{Fe55} for traversing points
in a square, layer by layer, on the other hand, we can prove that the optimal solutions for the three problems are
roughly the same in this setting.

\paragraph{Proof of Theorem~\ref{thm:roughly-the-same}.}
(i)~Recall that $\F$ is a maximal packing of unit disks in a large square~$S$.
Observe that $\opt_\orp(\F) \leq \opt_\ewrp(\F)$ (by problem definition). We start with the lower bounds.
  Let $A=10^{109}$ and $s$ denote the side length of $S$. We may assume without loss of generality
  that $n$ and $s$ are large enough, in particular that $s$ is a multiple of $A$. Thus $S$ can be partitioned into
  $s^2/A^2$ subsquares of side $A$. Since $\F$ is a maximal packing we have $n =\Theta(s^2)$. Let $T$ be
  an observation tour for $\F$. Since each of the $s^2/A^2$ subsquares of side $A$ has a hidden point, $T$ must visit
  every subsquare. Moreover, this implies that $\len(T) = \Omega(s^2/A^2)  - O(s) = \Omega(s^2/A^2) = \Omega(n)$.
By the inequality in the beginning, it follows that  $\opt_\ewrp(\F) \geq \opt_\orp(\F) = \Omega(n)$.

Next we show that $\opt_\tspn(\F)=\Omega(n)$. Let $T$ be a $\tspn$ tour for $\F$.
A standard disk packing argument (see~\cite[Prop.~1]{DM03} for details) yields $\len(T) \geq \pi (n-4)/4$,
whence $\opt_\tspn(\F)=\Omega(n)$, as claimed.

To prove the upper bounds, we first consider $\ewrp$.
Subdivide $S$ into horizontal strips (rectangles) $s \times 4$ and traverse the strips in a zig-zag manner,
say, from top to bottom. Finally return to the start position by following the boundary $\partial S$.
In every strip, move from left to right or from right to left and circle around the boundary of each disk
in increasing (resp., decreasing) order of the $x$-coordinates of their centers.  Since $\F$ is a maximal packing,
each disk-to-disk move is bounded from above by $O(1)$ in length and thus the total length is $O(n)$,
whence $\opt_\orp(\F) \leq \opt_\ewrp(\F) = O(n)$. Following the same (strip visitation) algorithm while visiting
each disk instead of circling around it, yields $\opt_\tspn(\F)=O(n)$, as claimed.

\smallskip
(ii)~The stated inequalities have been established when proving Theorems~\ref{thm:ORP-is-hard}
  and~\ref{thm:EWRP-is-hard}.
\qed

\section{Concluding remarks}  \label{sec:conclusion}

We conclude with a few open problems regarding the remaining gaps and the quality of approximation.

\begin{enumerate}  \itemsep 1pt

\item Is there a constant-ratio approximation algorithm for the shortest observation tour problem for families of  disjoint axis-aligned unit squares (or unit disks, or translates of a convex polygon)?

\item What approximations can be computed for the shortest observation tour for families of disjoint convex polygons?

\item Is there a constant-ratio approximation algorithm for the shortest external watchman route for families of disjoint axis-aligned unit squares (or unit disks or translates of a convex polygons)?

\item What approximations can be computed for the shortest external watchman route for families of disjoint convex polygons? Is the problem $\APX$-hard?

\end{enumerate}

%\bibliographystyle{plainurl}
%\bibliographystyle{abbrv}
%\bibliographystyle{splncs04}
%\bibliography{tours.bib}

\end{document}